%% file: main.tex
\documentclass[a4paper]{article}
\usepackage{graphicx}
\usepackage{amsmath}
\usepackage{amssymb}
\usepackage{amsthm}

\newtheorem{theorem}{Theorem}[section]

\newtheorem{lemma}[theorem]{Lemma}
\newtheorem{definition}[theorem]{Definition}

\newcommand{\name}[1]{\textsc{#1}}
 
\usepackage[all,british]{foreign}
\usepackage{microtype}

\usepackage{thmtools}
\usepackage{thm-restate}
\usepackage{xcolor}

\newcommand{\tbnd}{%
  \mathchoice%
  {\raisebox{4.5pt}{{\scriptsize$\circ$}}\mkern-2mu} 
  {\raisebox{4.5pt}{{\scriptsize$\circ$}}\mkern-2mu} 
  {\raisebox{4.1pt}{{\tiny$\circ$}}\mkern-3mu} 
  {\raisebox{4.5pt}{{\scriptsize$\circ$}}\mkern-2mu} 
}

\usepackage{schemata}

\usepackage[ruled,longend,vlined,linesnumbered]{algorithm2e}
\usepackage{wrapfig}
\usepackage{complexity}
\usepackage{paralist}
\newcommand{\widthm}[1]{\ensuremath{\mathop\mathbf{#1}}\xspace}
\newcommand{\tw}{\widthm{tw}}
\newcommand{\pw}{\widthm{pw}}
\newcommand{\td}{\widthm{td}}
\newcommand{\bound}{\ensuremath{\mathop{bd}}}    

\title{Width, depth and space}

\author{
  Li-Hsuan Chen\\
  \texttt{clh100p@cs.ccu.edu.tw}
  \and
  Felix Reidl\\
  \texttt{fjreidl@ncsu.edu}
  \and
  Peter Rossmanith\\
  \texttt{rossmani@cs.rwth-aachen.de}
  \and
  Fernando S\'anchez~Villaamil\\
  \texttt{fernando.sanchez@cs.rwth-aachen.de}
}

\def\Nesetril{Ne\v{s}et\v{r}il\xspace}

\begin{document}

\maketitle

\begin{abstract}
  \input{abstract}
\end{abstract}

\input{intro}

\input{preliminaries}

\input{lower-bounds}

\input{algorithm}

\input{conclusion}

\medskip
{\small
\noindent\textbf{Acknowledgements}
Research funded by DFG-Project RO 927/13-1 ``Pragmatic Parameterized
Algorithms''.
Li-Hsuan Chen was supported in part by MOST-DAAD Sandwich Program under grant no. 103-2911-I-194-506.
}


\bibliography{./biblio}

\end{document}

%% file: abstract.tex
The width measure \emph{treedepth}, also known as vertex ranking,
centered coloring and elimination tree height, is a well-established
notion which has recently seen a resurgence of interest. Since graphs
of bounded treedepth are more restricted than graphs
of bounded tree- or pathwidth, we are interested in the algorithmic
utility of this additional structure.

On the negative side, we show with a novel approach that every dynamic programming algorithm
on treedepth decompositions of depth~$t$ cannot solve \name{Dominating Set} 
with $O\big((3-\epsilon)^t \cdot \log n\big)$ space for any~$\epsilon > 0$.
This result implies the same space lower bound for dynamic programming algorithms 
on tree and path decompositions. We supplement this result by showing a space lower
bound of $O\big((3-\epsilon)^t \cdot \log n\big)$ for \name{$3$-Coloring} and
$O\big((2-\epsilon)^t \cdot \log n\big)$ for \name{Vertex Cover}. This
formalizes the common intuition that dynamic programming algorithms on
graph decompositions necessarily consume a lot of space and complements
known results of the \emph{time}-complexity of problems restricted to
low-treewidth classes~\cite{lokshtanov2011known}.

We then show that treedepth lends itself to
the design of branching algorithms. This class of algorithms has in
general distinct advantages over
dynamic programming algorithms:\begin{inparaenum}[a)]
\item They use less space than algorithms based on dynamic programming,
\item they are easy to parallelize and
\item they provide possible solutions before terminating, i.e.~they
  can be used to derive heuristics by using preliminary solutions.
\end{inparaenum}
Specifically, we design two novel algorithms for \name{Dominating Set}
on graphs of treedepth~$t$:
a pure branching algorithm that runs in time $t^{O(t^2)}\cdot n$
and uses space $O(t^3 \log t + t \log n)$ and a hybrid of
branching and dynamic programming that achieves a running time of $O(3^t \log t \cdot
n)$ while using $O(2^t t \log t + t \log n)$ space. Algorithms for
\name{$3$-Coloring} and \name{Vertex Cover} with space complexity~$O(t \cdot \log n)$
and time complexity~$O(3^t \cdot n)$ and $O(2^t\cdot n)$, respectively, are included
for completeness.


%% file: intro.tex
\section{Introduction}

The notion of \emph{treedepth} has been introduced several times in the
literature under several different names. It was first formally studied in the
context of matrix factorization as \emph{minimum elimination
trees}~\cite{schreiber1982new}; Katchalski \etal~studied the same notion under
the name of \emph{ordered colorings}~\cite{KMS95}; Bodlaender et al.~used the
term \emph{vertex ranking}~\cite{BDJ98}. Recently, Ossona de Mendez and
\Nesetril brought the same concept to the limelight in the guise of
\emph{treedepth} in their book
\textit{Sparsity}~\cite{NOdM12}.
Here we use that definition of treedepth which one we consider
easiest to exploit algorithmically: The treedepth $\td(G)$ of a graph
$G$ is the minimal height of a forest $F$ such that $G$ is a subgraph
of the closure of $F$. The closure of a forest is the graph resulting
from adding edges between every node and its ancestors, \ie making
every path from root to leaf into a clique. A \emph{treedepth
  decomposition} is a forest witnessing this fact.

Algorithmically, treedepth is interesting since it is structurally
more restrictive than pathwidth.
Treedepth bounds the pathwidth and treewidth of a graph, \ie $\tw(G)
\leq \pw(G) \leq \td(G) - 1 \leq \tw(G) \cdot \log n$, where $\tw(G)$
and $\pw(G)$ are the treewidth and pathwidth of an $n$-vertex graph $G$
respectively. Furthermore, a path decomposition can be easily
computed from a treedepth decomposition. 
Not only are there problems that are 
$W[1]$-hard when parameterized by treewidth or
pathwidth, but fpt when parameterized by treedepth~\cite{gutin2014structural};
low treedepth can also be exploited to count
the number of appearances of different substructures, such as
matchings and small subgraphs, much more
efficiently~\cite{demaine2014structural,furer2014space}. 

It turns out, however, that the `classical' approach of dynamic programming (DP) on graph
decompositions cannot possibly provide us with faster algorithms, even if we 
restrict ourselves to treedepth decompositions. Worse, we show in the following 
that the \emph{space} complexity is necessarily high. To clarify, we consider 
algorithms that take as input a tree-, path- or treedepth decomposition of width~$s$
and size~$n$ and satisfy the following constraints:
\begin{enumerate}
    \item They pass a single time over the decomposition in a bottom-up fashion;
    \item they use~$O(f(s) \log n)$ space; and
    \item they do not modify the decomposition, including re-arranging it.
\end{enumerate}
While these three constraints might look stringent, they include pretty much
all dynamic programming algorithm for hard optimization problems on tree or
path decompositions. For that reason, we will refer to this type of algorithms
simply as \emph{DP algorithms} in the following.

In order to show the aforementioned
space lower bounds, we introduce a simple machine model that models DP algorithms
on treedepth decompositions and construct 
superexponentially large Myhill-Nerode families
that imply lower bounds for \name{Dominating Set}, \name{Vertex Cover}, and 
\name{$3$-Colorability} in this algorithmic model. These lower bounds hold as
well for tree and path decompositions and align nicely with the space complexity
of known DP algorithms: for every $\epsilon > 0$, no DP algorithm on such
decomposition of width/depth $s$ can use space bounded by
$O\big((3-\epsilon)^s \cdot \log n\big)$ for \name{$3$-coloring} or
\name{Dominating Set} nor $O\big((2-\epsilon)^s\cdot \log n\big)$ for
\name{Vertex Cover}. While probably not surprising,
we consider a formal proof for what previously were just widely held
assumptions valuable. The provided framework should easily extend to
other problems.

Consequently, any algorithmic benefit of treedepth over pathwidth and
treewidth must be obtained by non-DP means. We demonstrate that treedepth
allows the design of branching algorithms whose space consumption grows only
polynomially in the treedepth and logarithmic in the input size. Such 
space-efficient algorithms are quite easy to obtain for \name{$3$-Coloring} and
\name{Vertex  Cover} with running time $O(3^t \cdot n)$ and $O(2^t \cdot n)$,
respectively, and space complexity $O(t \cdot \log n)$. Our main contribution
on the positive side here are two linear-fpt algorithms for \name{Dominating Set}
which use much more involved branching rules on treedepth
decompositions. The first one runs in time $t^{O(t^2)} \cdot n$ and
uses space $O(t^3 \log t + t \cdot \log n)$. Compared to simple
dynamic programming, the space consumption is improved considerably,
albeit at the cost of a much higher running time. For this
reason, we design a second algorithm that uses 
a hybrid approach of  branching and dynamic programming, resulting 
in a competitive running time of $O(3^t \log t \cdot n)$ and
space consumption $O(2^t t \log t + t \log n)$. 
Both algorithms are amenable to heuristic improvements (see
the Conclusion for a discussion), making them good candidates 
for real-world application. 

While applying branching to treedepth seems natural, it is unclear whether it
could be applied to treewidth or pathwidth. Very recent work by Drucker,
Nederlof, and Santhanam suggests that, relative to a collaps of the polynomial
hierarchy, \name{Independent Set} restricted to low-pathwidth graphs cannot be
solved by a branching algorithm in fpt-time~\cite{NoBranchPersonal}.

The idea of using treedepth to improve space consumption is not
entirely novel.  F{\"u}rer and Yu demonstrated that it is possible
count matchings using polynomial space in the size of the
input~\cite{furer2014space} and a parameter closely related to the
treedepth of the input. Their algorithm achieves a small memory
footprint by using the algebraization framework developed by
Loksthanov and Nederlof~\cite{lokshtanov2010saving}. This technique
was also used to develop an algorithm for \name{Dominating Set} which
runs in time $3^t \cdot \poly(n)$ (non-linear) and uses space
$O(t \cdot \log n)$~\cite{pilipczuk-space}. In our opinion, this type
of algorithm has two disadvantages: On the theoretical side, the
dependency on the input size is at least $\Omega(n)$. A dependence of
$O(\log n)$, as provided by dynamic programming, would be
preferable. On the practical side, using the \emph{Discrete Fourier
  Transform} makes it hard to apply common algorithm engineering
techniques, like \emph{branch \& bound}, which are available for
branching algorithms.

%% file: preliminaries.tex
\section{Preliminaries}

We write~$N_G[x]$ to denote the closed
neighbourhood of~$x$ in~$G$ and extends this notation to vertex
sets via~$N_G[S] := \bigcup_{x in S} N_G[x]$. Otherwise we
use standard graph-theoretic notation (see~\cite{Die10} for any
undefined terminology).
All our graphs are finite and simple and 
logarithms use base two. For reasons of space we defer some proofs
to the appendix and mark the respective claim with a~$\star$.
For sets~$A,B,C$ we write~$A \uplus B = C$ to express that~$A,B$
partition~$C$.

\begin{definition}[Treedepth]
  A \emph{treedepth decomposition} of a graph $G$ is a forest $F$ with
  vertex set $V(G)$, such that if $uv \in E(G)$ then either $u$ is an
  ancestor of $v$ in $F$ or vice versa. The \emph{treedepth} $\td(G)$
  of a graph~$G$ is the minimum height of any treedepth decomposition
  of $G$.
\end{definition}

\noindent
We assume that the input graphs are connected, which
allows us to presume that the treedepth decomposition is always a
tree. Furthermore, let $x$ be a node in some treedepth decomposition
$T$. We denote by $T_x$ the complete subtree rooted at $x$ and by
$P_x$ the set of ancestors of $x$ in $T$ (not including $x$). A \emph{subtree of~$x$}
refers to a subtree rooted at some child of~$x$.
The treedepth of a graph $G$ bounds its \emph{treewidth} $\tw(G)$ and
\emph{pathwidth} $\pw(G)$, \ie~$\tw(G)\leq \pw(G) \leq
\td(G)-1$~\cite{NOdM12}. For a definition of treewidth and pathwidth
see \eg Bodlaender~\cite{bodlaender1988dynamic}.

An \emph{$s$-boundaried graph}~$\tbnd G$ is a graph~$G$ with a set~$\bound(\tbnd G)
\subseteq V(G)$ of~$s$ distinguished vertices labeled~$1$ through~$s$,
called the \emph{boundary} of~$\tbnd G$. We will call vertices that are not
in~$\bound(\tbnd G)$ \emph{internal}.
By~$\tbnd \mathcal G_s$ we denote the class of all~$s$-boundaried graphs.
For $s$-boundaried graphs $\tbnd G_1$ and $\tbnd G_2$, we let the \emph{gluing}
operation~$\tbnd G_1 \oplus \tbnd G_2$ denote the $s$-boundaried graph obtained by first
taking the disjoint union of~$G_1$ and~$G_2$ and then unifying the boundary
vertices that share the same label\footnote{In the literature the result of gluing
is often an unboundaried graph. Our definition of gluing will be more convenient in
the following.}. 


%% file: lower-bounds.tex
\section{Space lower bounds for dynamic programming}

Lokshtanov, Marx, and Saurabh showed---assuming SETH---that
algorithms for \name{3-Coloring}, \name{Vertex Cover} and
\name{Dominating Set} on a tree decomposition of width~$w$ with
running time $O(3^w \cdot n)$, $O(2^w \cdot n)$ and $O(3^w w^2 \cdot
n)$, respectively, are basically optimal~\cite{lokshtanov2011known}. 
Their stated intent (as reflected in the title of the paper) was to
substantiate the common belief that known DP algorithms that solve
these problems where optimal. This is why we feel that a 
restriction to a certain type of algorithm is not necessarily inferior to
a complexity-based approach: indeed, most algorithms leveraging treewidth
\emph{are} dynamic programming algorithms or can be equivalently expressed as
such~\cite{bodlaender1988dynamic, bodlaender1997treewidth,
bodlaender2012fixed, borie1992automatic, borie2008solving,
scheffler1986dynamic}. Even before dynamic programming on tree-decompositions
became an important subject in algorithm design, similar concepts were already
used implicitly~\cite{bertele1973non}. The sentiment that the table size is
the crucial factor in the complexity of dynamic programming algorithms is
certainly not new (see \eg~\cite{van2009dynamic}), so it seems natural to
provide lower bounds to formalize this intuition. Our tool of choice will be a
family of boundaried graphs that are distinct under Myhill-Nerode equivalence.
The perspective of viewing graph decompositions as an `algebraic' expression
of boundaried graphs that allow such equivalences is well-established~\cite{bodlaender2012fixed,borie2008solving}.

First of all, we need to establish what we mean by \emph{dynamic programming}.
DP algorithms on graph decompositions work by visiting the bags/nodes of the
decomposition in a bottom-up fashion (a post-order depth-first traversal),
maintaining DP tables to compute a solution. For decision problems, these
algorithms only need to keep at most $\log n$ tables in memory at any given
moment. We propose a machine model with a read-only tape for the input that
can only be traversed once, which only accepts as input decompositions
presented in a valid order. This model suffices to capture known dynamic
programming algorithms on path, tree and treedepth decompositions. More
specifically, given a decision problem on graphs $\Pi$ and some well-formed
instance $(G,\xi)$ of~$\Pi$ (where~$\xi$ encodes the non-graph part of the
input), let $T$ be a tree, path or treedepth decomposition of $G$ of
width/depth~$k$. We fix an encoding~$\hat T$ of~$T$ that lists the separators provided
by the decomposition in the order they are normally visited in a dynamic
programming algorithm (post-order depth-first traversal of the bag/nodes of a
tree/path/treedepth decomposition) and additionally encodes the edges of~$G$
contained in a separator using $O(k \log k)$ bits per bag. Then
$(k,\hat T,\xi)$ is a well-formed instance of the \emph{DP decision problem
$\Pi_{\text{DP}}$}. Pairing DP decision
problems with the following machine model provides us with a way to model
DP computation over graph decompositions.

\begin{definition}[Dynamic programming TM]
  A \emph{DPTM} $M$ is a Turing machine with an input
  read-only tape, whose head moves only in one
  direction, and a separate working tape.
  It accepts as inputs well-formed instances of some DP
  decision problem.
\end{definition}

\noindent
Any single-pass dynamic programming algorithm that solves a DP decision problem
on tree, path or treedepth decompositions of width/depth~$k$ using tables of size~$f(k)$
that does not re-arrange the decomposition 
can be translated into a DPTM with a working tape of size $O(f(k) \cdot \log n)$.
This model does not suffice to rule
out algebraic techniques, since this technique, like branching, requires to
visit every part of the decomposition many times~\cite{furer2014space};
or algorithms that preprocess the decomposition first to find a suitable
traversal strategy.

The following notion of a \emph{Myhill-Nerode family} will provide us 
with the machinery to prove space lower-bounds for DPTMs and hence common
dynamic programming algorithms. Recall that~$\tbnd \mathcal G_s$ denotes
the class of all~$s$-boundaried graphs.

\begin{definition}[Myhill-Nerode family]
  A set~$\mathcal H \subseteq \tbnd \mathcal G_s \times \mathbf N$
  is a \emph{$s$-Myhill-Nerode family} for a DP-decision problem~$\Pi_{\text{DP}}$ if the following holds:
  \begin{enumerate}
    \item Every $s$-boundaried graph~$\tbnd H \in \tbnd \mathcal G_s$
          appears at most once in~$\mathcal H$ and if it does appear its size is bounded
          by~$|\tbnd H| \leq |\mathcal H| \polylog |\mathcal H|$.
    \item For every subset~$\mathcal I \subseteq \mathcal H$ there exists an
          $s$-boundaried graph $\tbnd G_{\mathcal I}$ with 
          $|\tbnd G_{\mathcal I}| \leq |\mathcal H| \polylog |\mathcal H|$ and an
          integer~$p_{\mathcal I}$ such that for every~$(H,q) \in \mathcal H$ it holds
          that
          \[
            (\tbnd G_{\mathcal I} \oplus \tbnd H, p_{\mathcal I}+q ) \in \Pi_{\text{DP}} 
             \iff (H,q) \not \in \mathcal I.
          \]
  \end{enumerate}
\end{definition}

\noindent
We define the \emph{size} of a Myhill-Nerode family~$\mathcal H$ as~$|\mathcal H|$
and its \emph{treedepth}  as
\[
  \td(\mathcal H) = \max_{(\tbnd H,\cdot) \in \mathcal H, \mathcal I \subseteq \mathcal H} 
  \td( \tbnd G_{\mathcal I} \oplus \tbnd H \oplus \tbnd K_s ).
\]
We can similarly define the treewidth and pathwidth of a family. Note that by
turning the boundary in the above definition into a clique we essentially
measure the treedepth under the assumption that all boundary vertices appear
as a path of length~$s$ on the top of the decomposition. The following lemma
still holds if we replace `treedepth' by `pathwidth' or `treewidth'.

\begin{lemma}\label{lemma:no-dptm}
  Let $\epsilon>0, c > 1$ and $\Pi$ be a DP decision problem such that for
  every $s$ there exists an $s$-Myhill-Nerode family $\mathcal H$ for $\Pi$ of size
  ${c^s/\poly(s)}$ and depth~$\td(\mathcal H) = s + o(s)$. Then no DPTM can
  decide $\Pi$ using space~$O((c-\epsilon)^k\log n)$, where $n$ is the size
  and $k$ the depth of the treedepth decomposition given as input.
\end{lemma}
\begin{proof}
  Assume to the contrary that such a DPTM $M$ exists. Fix~$s$ and consider any
  subset~$\mathcal I \subseteq \mathcal H$ of the $s$-Myhill-
  Nerode family~$\mathcal H$ of~$\Pi$. By definition and our above
  assumptions, all graphs in~$\mathcal H$ 
  and the graph~$\tbnd G_{\mathcal I}$ have size at most
  \[
    |\mathcal H| \polylog |\mathcal H| = c^s \cdot \poly(s).
  \]
  For every $s$-boundaried graph~$\tbnd H$ contained in~$\mathcal H$, there
  exist treedepth decompositions for~$\tbnd G_{\mathcal I} \oplus \tbnd H$ of
  depth at most~$s + o(s)$ such that the boundary vertices~$\bound(\tbnd G_{
  \mathcal I})$ appear on a path of length~$s$ on the top of the decomposition.
  Hence, we can fix a treedepth decomposition~$T_{\mathcal I}$ of~$G_{\mathcal I}$ with
  exactly that property and chose a treedepth decomposition of~$\tbnd
  G_{\mathcal I} \oplus \tbnd H$ that contains~$T_{\mathcal I}$.
  Moreover, we chose an encoding of the resulting decomposition that lists
  the nodes of~$T_{\mathcal I}$ first.

  There are~$2^{|\mathcal H|} = 2^{c^s/\poly(s)}$ choices
  for~$\mathcal I$ but~$M$ only uses~$(c-\epsilon)^{s+o(s)} \cdot \poly(s +
  o(s)) = c^{s-o(s)}/\poly(s)$ space. Hence by the pigeonhole principle there exist
  graphs~$\tbnd G_{\mathcal I}, \tbnd G_{\mathcal J}$, $\mathcal I \neq
  \mathcal J \subseteq \mathcal H$ for which~$M$ is in the same state and
  has the same working-tape content after reading the nodes of the respective
  decompositions~$T_{\mathcal I}$ and~$T_{\mathcal J}$. Choose 
  $(\tbnd H,q) \in \mathcal I \bigtriangleup \mathcal J$. By definition, we have that
  \[
    (\tbnd G_{\mathcal I} \oplus \tbnd H, p_{\mathcal I} + q ) \in \Pi
    \iff 
    (\tbnd G_{\mathcal J} \oplus \tbnd H, p_{\mathcal J} + q ) \not \in \Pi
  \]
  but~$M$ will either reject or accept both inputs. Contradiction.
\end{proof}

\noindent 
The following space lower bounds all follow the same basic construction.
We define a problem-specific `state' for the vertices of a boundary set~$X$
and construct two boundaried graphs for it: one graph that enforces this
state in any (optimal) solution of the respective problem and one graph that
`tests' for this state by either rendering the instance unsolvable or increasing
the costs of an optimal solution. We begin with the easiest construction: 
a Myhill-Nerode family for \name{$3$-Coloring}.

\begin{figure}[b]
  \vspace{-15pt}
  \begin{center}
    \includegraphics[width=0.53\textwidth]{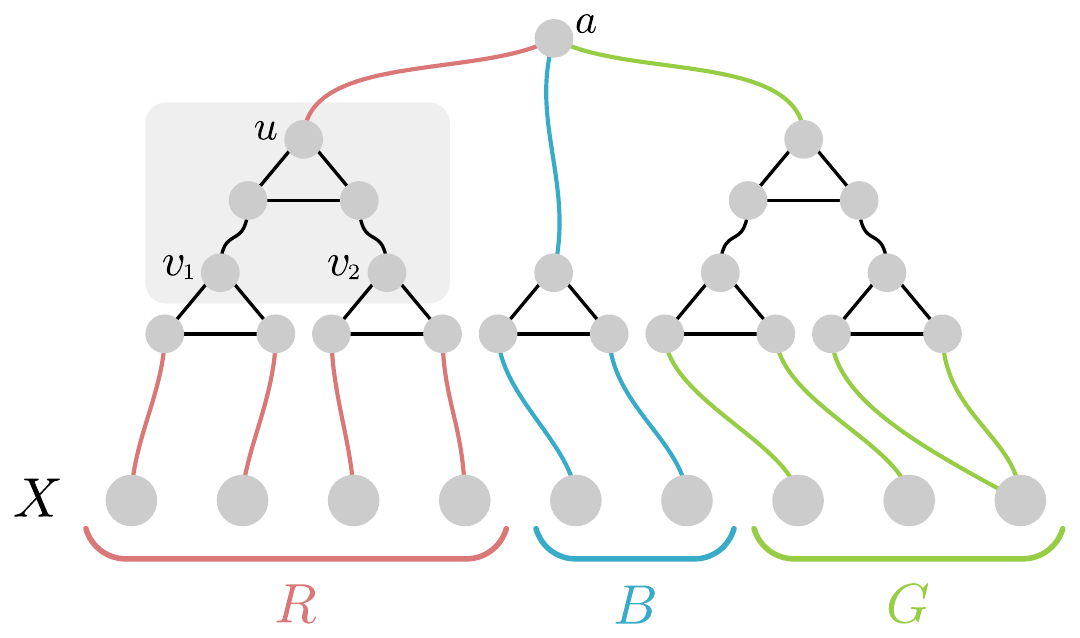}
  \end{center}
 \vspace{-15pt}
 \caption{The gadget~$\protect\tbnd \Gamma_{\! \mathcal X}$ for~$\mathcal X = \{R,G,B\}$.\label{fig:coloring-gadget}}
\end{figure}

\begin{theorem}\label{thm:lower-bound-3-coloring}
  For every $\epsilon>0$, no DPTM solves
  \name{3-Coloring} on a tree, path or treedepth decomposition of
  width/depth $k$ with space bounded by $O((3-\epsilon)^k\log n)$.
\end{theorem}
\begin{proof}
  For any~$s$ we construct an $s$-Myhill-Nerode family~$\mathcal H$ as
  follows. Let~$X$ be the $s$ boundary vertices of all the boundaried graphs
  in the following. Then for every three-partition~$\mathcal X = \{R,G,B\}$
  of~$X$ we add a boundaried graph~$\tbnd H_{\mathcal X}$ to the
  family~$\mathcal H$ by taking a single triangle~$v_R,v_G,v_B$ and connect
  the vertices~$v_C$ to all vertices in~$X \setminus C$ for $C \in \{R,G,B\}$.
  Since instances of three-coloring do not need any additional parameter,
  we ignore this part of the construction of~$\mathcal H$ and implicitly assume
  that every graph in~$\mathcal H$ is paired with zero.

  To construct the graphs~$G_{\mathcal I}$ for~$\mathcal I \subset \mathcal
  H$, we will employ the
  \emph{circuit gadget}~$v_1,v_2,u$
  highlighted in Figure~\ref{fig:coloring-gadget}: note that if $v_1,v_2$
  receive the same color, then necessarily~$u$ must be colored the same. In
  every other case, the color of $u$ is arbitrary. Now for every three-
  partition~$\mathcal X = \{R,G,B\}$ of~$X$, we construct a \emph{testing
  gadget}~$\tbnd \Gamma_{\mathcal X}$ as follows: for~$C \in \{R,G,B\}$, we
  arbitrarily pair the vertices in $C$ and connect them via the circuit gadget
  (as $v_1,v_2$). If~$|C|$ is odd, we pair some vertex of~$C$ with itself. We
  then repeat the construction with all the $u$-vertices of those gadgets,
  resulting in a hierarchical structure of depth~$\sim \log |B_i|$ (\cf
  Figure~\ref{fig:coloring-gadget} for an exemplary construction). Finally, we
  add a single vertex~$a$ and connect it to the top vertex of the three
  circuits. Note that by the properties of the circuit gadget, the
  graph~$\tbnd \Gamma_{\mathcal X}$ is three-colorable iff the coloring of~$X$
  does \emph{not} induce the partition~$\cal X$. In particular, the
  graph~$\tbnd \Gamma_{\mathcal X} \oplus \tbnd H_{\mathcal X'}$ is 
  three-colorable iff~$\mathcal X \neq \mathcal X'$.

  Now for every subset~$\mathcal I \subseteq \mathcal H$ of graphs from the
  family, we define the graph~$
    \tbnd G_{\mathcal I} = \bigoplus_{\tbnd H_{\mathcal X} \in \mathcal H} 
    \Gamma_\mathcal{X}$.
  By our previous observation, we have
  that for every~$\tbnd H_\mathcal{X} \in \mathcal H$ the graph~$\tbnd G_{\mathcal I}
  \oplus \tbnd H_\mathcal{X}$ is 
  three-colorable iff~$\tbnd H_\mathcal{X} \not \in \mathcal I$. Furthermore,
  every composite graph has treedepth at most~$s + 4   \lceil \log s \rceil +
  4$ as witnessed by a decomposition whose top~$s$ vertices are the
  boundary~$X$. The graphs~$\tbnd G_{\mathcal I}$ for every~$\mathcal I
  \subseteq
  \mathcal H$ have size at most~$3^s \cdot 6s$---we conclude that $\mathcal H$
  is an $s$-Myhill-Nerode family of size~${3^s / 6}$ (the factor~$1/6$ accounts
  for the~$3!$ permutations of the partitions) and the claim follows from 
  Lemma~\ref{lemma:no-dptm}.
\end{proof}

\noindent
Surprisingly, the construction to prove a lower bound for
\name{Vertex Cover} is very similar to the one for
\name{$3$-coloring}. For that reason we defer the proof to the 
appendix.

\begin{theorem}
  For every $\epsilon>0$, no DPTM solves
  \name{Vertex Cover} on a tree, path or treedepth decomposition of
  width/depth $k$ with space bounded by $O((2-\epsilon)^k\log n)$.
\end{theorem}
\begin{proof}
  Let $s=|X|$ be the size of the boundary.
  For every subset~$A \subseteq X$ we construct a
  graph~$\tbnd  H_{\!A}$ for~$\cal H$ which consists of the boundary and a
  matching to~$A$ plus~$s - |A|$ isolated $K_2$s as padding
  and add~$(H_{\!A}, s)$ to~$\mathcal H$.

  Consider~$\mathcal I \subseteq \mathcal H$. We will again use
  the circuit gadget highlighted in Figure~\ref{fig:coloring-gadget}
  to construction~$\tbnd G_{\mathcal I}$. We assign a budget of two
  vertices for every such circuit gadget. Note that if one of~$v_1, v_2$ 
  is in the vertex cover, then we can include the top vertex~$u$ into the 
  cover without exceeding this budget. Otherwise, $u$ cannot be included
  within the budget. For a set~$A \subseteq X$ we construct the
  testing gadget~$\tbnd \Gamma_{\!A}$ by connecting the vertices \
  of~$X\setminus A$
  pairwise via the circuit gadget (using an arbitrary pairing and 
  potentially pairing a leftover vertex with itself). As in the
  proof of Theorem~\ref{thm:lower-bound-3-coloring}, we repeat this
  construction on the respective $u$-vertices of the just added
  circuits and iterate until we have added a single circuit on the
  very top. Let us denote the topmost $u$-vertex in this construction
  by~$u'$. Let~$\lambda$ be the number of circuits added in this
  fashion, then the total budget for~$\tbnd \Gamma_{\!A}$ is~$2\lambda$.
  Note that if a vertex of~$X \setminus A$ is inside a vertex cover,
  then the gadget~$\tbnd \Gamma_{\!A}$ has a vertex cover of size~$2\lambda$
  that includes its top vertex~$u'$. Otherwise, no vertex cover of $\tbnd \Gamma_{\!A}$
  of size~$2\lambda$ includes~$u'$.

  We construct~$\tbnd G_{\mathcal I}$ by 
  taking~$\oplus_{\tbnd H_{\!A} \in \mathcal I} \tbnd \Gamma_{\!A}$ and adding a 
  single vertex~$a$ that connects to all $u'$-vertices of the 
  gadgets $\{ \tbnd \Gamma_{\!A} \}_{\tbnd H_{\!A} \in \mathcal I}$.
  For simplicity, we pad out the graph with isolated~$K_2$s
  to ensure that every graph $\tbnd G_{\mathcal I}, \mathcal I \subseteq \mathcal H$
  constructed in this way has a minimum vertex cover of the same size~$\ell$. 

  We claim that~$\tbnd G_{\mathcal I} \oplus \tbnd H_{\!A}$ has a vertex
  cover of size~$\ell$ iff $\tbnd H_{\!A} \not \in \mathcal I$. If~$H_{\!A} \not \in \mathcal I$,
  then every gadget~$\Gamma_{A'}$ that comprises~$\tbnd G_{\mathcal I}$
  has~$A' \neq A$. By the above observation, every gadget~$\Gamma_{A'}$
  can therefore be covered with~$2\lambda$ vertices including the
  vertex~$u'$; thus all edges incident to the vertex~$a$ are covered
  and we stay within budget. If $H_{\!A} \in \mathcal I$
  then~$\tbnd G_{\mathcal I}$ contains the gadget~$\tbnd \Gamma_{\!A}$ and
  as observed we cannot include its $u'$-vertex into the budgeted~$2\lambda$
  vertices. Hence, the edge~$u'a$ cannot be covered within the total
  budget~$\ell$ and we need to invest~$\ell+1$ vertices to obtain
  a vertex cover. Letting~$p_{\mathcal I} = \ell - s$, we conclude
  that~$\mathcal H$ is a $s$-Myhill-Nerode family of size~$2^s$ and
  depth~$\td(\mathcal H) = s + o(s)$.
\end{proof}


\begin{figure}[thb]
  \vspace{-5pt}
  \begin{center}
    \includegraphics[width=0.7\textwidth]{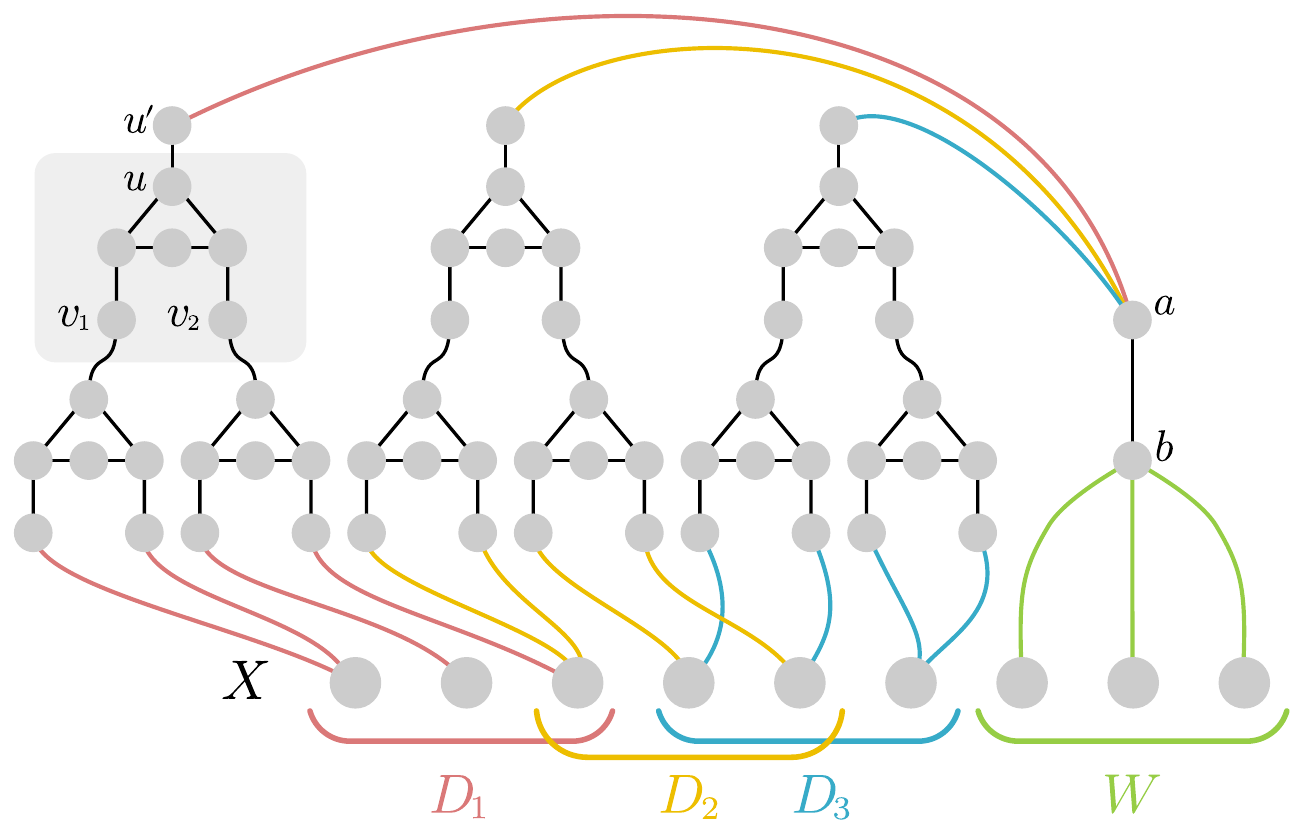}
  \end{center}
 \vspace{-15pt}
 \caption{The gadget~$\protect\tbnd \Gamma_{W}$ for~$\mathcal D_W = \{S_1,S_2,S_3\}$. Padding-vertices are not included.\label{fig:domset-gadget}}
  \vspace{-10pt}
\end{figure}

\begin{theorem}
  For every $\epsilon>0$, no DPTM solves
  \name{Dominating Set} on a tree, path or treedepth decomposition of
  width/depth $k$ with space bounded by $O\big((3-\epsilon)^k\log n\big)$.
\end{theorem}
\begin{proof}
  For any~$s$ dividable by three we construct an $s$-Myhill-Nerode
  family~$\mathcal H$ as follows. Let~$X$ be the $s$ boundary vertices of all
  the boundaried graphs in the following. Then for every three-
  partition~$\mathcal X = (B,D,W)$ of~$X$ into sets of size $s/3$, we
  construct a graph~$H_{\mathcal X}$ by connecting two pendant vertices to
  every vertex in $B$, connecting every vertex in $D$ to a vertex with two
  pendant vertices and leaving $W$ untouched. Intuitively, we want the
  vertices of~$B$ to be in any minimal dominating set, the vertices in~$D$ to
  be dominated from a vertex in~$H_{\mathcal X}$ and the vertices in~$W$ to be
  dominated from elsewhere. We add every pair $(H_{\mathcal X},s/3+1)$ to
  $\mathcal H$.

  For a subset~$\mathcal I \subseteq \mathcal H$ let~$\mathcal D_W = \{ D \mid
  H_{X\setminus (D\cup W),D,W} \in \mathcal I_W \}$ be all $D$-sets that
  appear together with a fixed~$W \subseteq X$ in~$\mathcal I$.
  We construct the graph~$\tbnd
  G_{\mathcal I}$ using the \emph{circuit gadget}~$v_1,v_2,u$ highlighted in
  Figure~\ref{fig:domset-gadget}: note that if $v_1,v_2$ are not dominated,
  then there is no dominating set of the gadget of size two that contains $u$.
  If one of~$v_1,v_2$ is dominated then such a dominating set containing $u$
  exists. For every~$W \subset X$ with~$|W| = s/3$
  construct a \emph{testing gadget}~$\tbnd \Gamma_W$ using
  the circuit gadgets as follows. 
  If~$\mathcal D_W$ is non-empty we construct~$\tbnd \Gamma_W$ as follows.
  For \emph{every} set $D \in \mathcal D_W$ we construct the gadget $\tbnd \Lambda_D$
  by arbitrarily pairing the vertices in $D$ and connecting them via the
  circuit gadget.  If~$|D|$ is odd, we pair some vertex of~$D$ with itself. We
  then repeat the construction with all the $u$-vertices of those gadgets,
  resulting in a hierarchical structure of depth~$\sim \log |D|$. 
  To finalize the construction of~$\tbnd \Lambda_D$ we take
  the $u$-vertex of the last layer and connect it to a vertex $u'$.
  Let in the following~$\lambda$ be the number of circuits we used to construct
  such a~$\tbnd \Lambda_D$-gadget (this quantity only depends on~$s$).
  This concludes the construction of~$\tbnd \Lambda_D$.
  If~$\mathcal D_W$ is empty, then
  we let~$\tbnd \Gamma_W$ be a~$K_2$ with one vertex connected to all vertices
  in~$W$ plus ${2s/3 \choose s/3} 2\lambda$ isolated padding-vertices.
  Otherwise we obtain~$\tbnd
  \Gamma_W$ by taking the graph~$\bigoplus_{D \in \mathcal D}
  \! \tbnd \Lambda_D$ and adding two additional vertices~$a,b$ as well as 
  $\big({2s/3 \choose s/3} - |\mathcal D_W|\big) 2 \lambda$
  isolated vertices for padding. The vertex~$a$ is connected to all~$u'$
  vertices of all the gadgets~$\{ \tbnd \Lambda_D \}_{D \in \mathcal D_W}$ and
  the vertex $b$ is connected to~$\{a\} \cup W$ (\cf
  Figure~\ref{fig:domset-gadget} for an exemplary construction).

  We assign the budget $\alpha = {2s/3 \choose s/3} 2\lambda + 1$ for the
  gadget~$\tbnd \Gamma_W$. If~$\mathcal D_W = \varnothing$ this budget is sufficient
  to include all padding-vertices and the endpoint of the $K_2$ that connects to~$W$.
  Hence, we can dominate~$W$ as well as~$\tbnd \Gamma_W$ within this budget.
  For non-empty~$\mathcal D_W$, after accounting for the padding-vertices,
  we are left with a budget of~$2\lambda$ for each gadget~$\tbnd \Lambda_D$,
  $D \in \mathcal D_W$ that comprise~$\tbnd \Gamma_W$.
  While at least~$2\lambda$ vertices will always be necessary to dominate~$\tbnd \Lambda_D$, only
  if at least one vertex of~$D$ is contained in the dominating set we can
  dominate the $u'$-vertex of~$\Lambda_D$ within budget. In the case
  that no vertex of~$D$ is in the dominating set $2\lambda$ vertices only suffice
  to dominate all vertices of~$\tbnd \Lambda_D$ except~$u'$.
  Hence, the gadget~$\tbnd \Gamma_W$ can dominate itself \emph{and}~$W$
  within the budget~$\alpha$ exactly when for all~$D \in \mathcal D_W$,
  least one vertex of~$D$ is in the dominating set.
  Otherwise, we need to include~$a$ into the dominating set and hence our budget
  does not suffice to include~$b$ in order to dominate~$W$.
  Finally we define for every~$\mathcal I \subseteq \mathcal H$ the graph
   $\tbnd G_{\mathcal I} = \bigoplus_{W \subset X, |W| = s/3} \! \tbnd \Gamma_W$.

  Let us now show that our boundaried graphs work as intended and calculate
  the appropriate parameters~$p_{\mathcal I}$. Consider any graph~$\tbnd
  H_{B_0,D_0,W_0} \in \mathcal H$ where again~$B_0,D_0,W_0$ is a partition
  of~$X$ into sets of size~$s/3$ and the graph~$\tbnd G_{\mathcal I}$ for
  any~$\mathcal I \subseteq \mathcal H$. We show that~$\tbnd H_{B_0,D_0,W_0}
  \oplus G_{\mathcal I}$ has a dominating set of size~${s \choose s/3} \alpha
  + s/3 + 1$ iff $\tbnd H_{B_0,D_0,W_0} \not \in \mathcal I$ (otherwise a minimal
  dominating set will exceed his budget by one).

  We use the sets~$\mathcal D_W, W
  \subseteq X$ as defined previously. 
  First, assume that~$\mathcal D_{W_0} = \varnothing$, that is,
  for every set~$B,D$ we have that~$H_{B,D,W_0} \not \in \mathcal I$ and
  in particular~$H_{B_0,D_0,W_0} \not \in \mathcal I$. As observed above,
  the gadget~$\Gamma_{W_0}$ can dominate itself and~$W_0$ with a budget of~$\alpha$.
  All other gadgets~$\Gamma_W, W \neq W_0$ do not need to dominate their
  respective~$W$-sets and can therefore include their $a$-vertices, accordingly
  they can all dominate their internal vertices within the assigned budget~$\alpha$.
  Including the~$s/3$ vertices of~$B_0$ as well as the one vertex in~$H_{B_0,D_0,W_0}$
  used to dominate~$D_0$ this correctly
  sums up to a dominating set of size~${s \choose s/3} \alpha + s/3 + 1$.
  Next, assume that~$\mathcal D_{W_0} \neq \varnothing$ and~$D_0 \not \in \mathcal D_{W_0}$, \ie 
  again~$H_{B_0,D_0,W_0} \not \in \mathcal I$. Therefore, for every set~$D \in D_{W_0}$
  we have that~$D \cap B_0 \neq \varnothing$. Since we include~$B_0$ in our dominating
  set, the gadgets~$\{\tbnd \Lambda_D\}_{D \in D_{W_0}}$ that comprise~$\tbnd \Gamma_{W_0}$ can
  all dominate their respective~$u'$-vertices within the assigned budet of~$2\lambda$.
  Hence, $\tbnd \Gamma_{W_0}$ can dominate~$W_0$ using the assigned budget of~$\alpha$.
  All other gadgets~$\Gamma_W, W \neq W_0$ again stay within budget as observed above
  and we obtain in total a dominating set of size~${s \choose s/3} \alpha + s/3 + 1$.
  Finally, consider the case that~$D_0 \in \mathcal D_{W_0}$, \ie 
  $H_{B_0,D_0,W_0} \in \mathcal I$. Then the gadget~$\Lambda_{D_0}$ contained
  in~$\Gamma_{W_0}$ cannot dominate its~$u'$ vertex within a budget of~$2\gamma$,
  hence we need to include either~$u'$ or~$a$. This depletes the budget of~$\Gamma_{W_0}$
  and hence we need in total~${s \choose s/3} \alpha + s/3 + 2$ vertices to
  dominate~$\tbnd H_{B_0,D_0,W_0} \oplus G_{\mathcal I}$. 

  Choosing~$p_{\mathcal I} = {s \choose s/3} \alpha$ completes the construction
  of~$(\tbnd \Gamma_{\mathcal I}, p_{\mathcal I})$ and we conclude that~$\mathcal H$
  is an $s$-Myhill-Nerode family of size~$\Omega(3^s / s)$. For~$s$ indivisble by
  three we take the next smaller integer~$s'$ divisable by three and use the
  $s'$-family as the $s$-family. It is easy to confirm that the depth of~$\mathcal H$
  is~$\td(\mathcal H) = s + o(s)$ and the theorem follows from Lemma~\ref{lemma:no-dptm}.
\end{proof}


%% file: algorithm.tex
\section{\name{Dominating Set} using \boldmath${O(t^3 \log t + t
    \log n)}$ space}\label{sec:branching}

That branching might be a viable algorithmic design strategy for
low-treedepth graphs can easily be demonstrated for problems like
\name{$3$-Coloring} and \name{Vertex Cover}: we simply branch on the
topmost vertex of the decomposition and recurse into (annotated) subinstances.
For \name{$c$-Coloring}, this leads to an algorithm with running time
$O(c^t \cdot n)$ and space complexity $O(t \log n)$.
Since it is possible to perform a depth-first traversal
of a given tree using only $O(\log n)$
space~\cite{lindell1992logspace}, the space consumption of this
algorithm can be easily improved to $O(t + \log n)$. Similarly,
branching solves \name{Vertex Cover} in time~$O(2^t \cdot n)$
and space~$O(t \log n)$. 

The task of designing a similar algorithm for \name{Dominating Set}
is much more involved. Imagine branching on the topmost vertex of the
decomposition: while the branch that includes the vertex into the
dominating set produces a straightforward recurrence into annotated
instances, the branch that excludes it from the dominating set needs
to decide \emph{how} that vertex should be dominated. The algorithm
we present here proceeds as follows: We first guess whether the current node~$x$ is
in the dominating set or not. Recall that~$P_x$ denotes the nodes
of the decomposition that lie on the unique path from~$x$ to the root
of the decomposition (and~$x \not \in P_x$). We iterate over every possible
partition $S_1 \uplus \dots \uplus S_l = P_x \cup \{x\}$ 
into $l \leq t$ sets of $P_x \cup \{x\}$. 
The semantic of a block~$S_i$ is that we want every element~$S_i$
to be dominated exclusively by nodes from a specific subtree of~$x$.
A recursive call on a child $y$ of $x$, together with an element
of the partition $S_i$, will return the size of a dominating set which
dominates~$V(T_y) \cup S_i$. The remaining issue is how these specific
solutions for the subtrees of~$x$ can be combined into a solution in
a space-efficient manner. To that end, we first compute
the size of a dominating set for $T_y$ itself and use this as baseline 
cost for a subtree~$T_y$. For a block~$S_i$ of a partition of~$P_x$,
we can now compare the cost of dominating~$V(T_y) \cup S_i$ against
this baseline to obtain overhead cost of dominating~$S_i$ using
vertices from~$T_y$. Collecting these overhead costs in a table for
subtrees of~$x$ and the current partition, we are able to apply
certain reduction rules on these tables to reduce their size to at most
$t^2$ entries. Recursively choosing the best partition then yields the
solution size using only polynomial space in $t$ and logarithmic in
$n$. Formally, we prove the following:

\begin{theorem}\label{thm:poly-space-domset}
  Given a graph $G$ and a treedepth decomposition $T$ of $G$,
  Algorithm~\ref{alg:domset-polyspace} finds the size of a minimum
  dominating set of $G$ in time $t^{O(t^2)} \cdot n$ using
  $O(t^3 \log t + t \log n)$ bits.
\end{theorem}

\noindent
We split the proof of Theorem~\ref{thm:poly-space-domset} into lemmas
for correctness, running time and used space. 

\begin{lemma}
  Algorithm~\ref{alg:domset-polyspace} called on a graph $G$, a
  treedepth decomposition $T$ of $G$, the root $r$ of $T$ and $P = D =
  \varnothing$ returns the size of a minimum dominating set of $G$.
\end{lemma}
\begin{proof}
  If we look at a minimal dominating set $S$ of $G$ we can charge
  every node in $V(G) \setminus S$ to a node from $S$ that dominates
  it. We are thus allowed to treat any node in $G$ as if it was
  dominated by a single node of $S$.
  We will prove this lemma by induction, the inductive hypothesis
  being that a call on a node $x$ with arguments $D = S \cap P_x$ and
  $P \subseteq P_x$ being the set of nodes dominated from nodes in
  $T_x$ returns $|S \cap V(T_x)|$.

  It is clear that the algorithm will call itself until a leaf is
  reached. Let $x$ be a leaf of $T$ on which the function was
  called. We first check the condition at line~\ref{line:take-leaf},
  which is true if either $x$ is not dominated by a node in $D$ or if
  some node in $P$ is not yet dominated. In this case we have no
  choice but to add $x$ to the dominating set. Three things can
  happen: $P$ is not fully dominated, which means that it was not
  possible under these conditions to dominate $P$, in which case we
  correctly return $\infty$, signifying that there is no valid
  solution. Otherwise we can assume $P$ is dominated and we return $1$
  if we had to take $x$ and $0$ if we did not need to do so. Thus the
  leaf case is correct.

  We assume now $x$ is not a leaf and thus we reach
  line~\ref{line:non-leaf}. We first add $x$ to $P$, since it can only
  be dominated either from a node in $D$ or a node in $T_x$. Nodes in
  $T_x$ can only be dominated by nodes from $V(T_x) \cup P_x$. We
  assume by induction that $D = S \cap P_x$ and that $P$ only contains
  nodes which are either in $S$ or dominated from nodes in
  $T_x$. Algorithm~\ref{alg:domset-polyspace} executes the same
  computations for $D$ and $D \cup \{x\}$, representing not taking and
  taking $x$ into the dominating set respectively. We must show that
  the set $P$ for the recursive calls is correct. There exists a
  partition of the nodes of $P$ not dominated by $D$ (respectively $D
  \cup \{x\}$) such that the nodes of every element of the partition
  are dominated from a single subtree $T_y$ where $y$ is a child of
  $x$. The algorithm will eventually find this partition on
  line~\ref{line:partition-loop}. The baseline value, i.e.~the size of
  a dominating set of $T_y$ given that the nodes in $D$ (respectively
  $D \cup \{x\}$) are in the dominating set, gives a lower bound for
  any solution. In the lists in $L$ and $L'$ we keep the extra cost
  incurred by a subtree $T_y$ if it has to dominate an element of the
  partition. 
  We only need to keep the best $t$ values
  for every $S_i$: Assume that it is optimal to dominate $S_i$ from
  $T_y$ and there are $t+1$ subtrees induced on children of $x$ whose
  extra cost over the baseline to dominate $S_i$ is strictly smaller
  than the extra cost for $T_y$. At least one of these subtrees
  $T_{y'}$ is not being used to dominate an element of the
  partition. This means we could improve the solution by letting $T_y$
  dominate itself and taking the solution of $T_{y'}$ that also
  dominates $S_i$. Keeping $t$ values for every element in the
  partition suffices to find a minimal solution, which is what
  $find\_min\_solution$ does. Since with lines~\ref{line:min-1}
  and~\ref{line:min-2} we take the minimum over all possible
  partitions and taking $x$ into the dominating set or not, we get
  that by inductive assumption the algorithm returns the correct
  value. The lemma follows since the first call to the algorithm with
  $D = P = \varnothing$ is obviously correct.
\end{proof}

\def\nothing{\varnothing}
\begin{algorithm}[tbh]
  
  \footnotesize
  \caption{domset{\label{alg:domset-polyspace}\sc }}
  \KwIn{A graph $G$ and a treedepth decomposition $T$ of $G$, a node
    $x$ of $T$ and sets $P,D \subseteq V(G)$.}  
  \KwOut{The size of a minimum Dominating Set.} \BlankLine

  \If{$x$ is a leaf in $T$}{
    \lIf{$x \notin N_G[D]$ or $P \not\subseteq N_G[D]$\label{line:take-leaf}}{
      $D := D \cup \{x\}$
    }
    \lIf{$P \not\subseteq N_G[D]$}{
      \Return $\infty$
    }\lElseIf{$x \in D$}{
      \Return $1$
    }\lElse{
      \Return $0$
    }
  }
  $result := \infty$\label{line:non-leaf}\;
  $P := P \cup \{x\}$\;
  \ForEach{partition $S_1 \uplus \dots \uplus S_l$ of $P$\label{line:partition-loop}}{
    $L := |P|$-element array of ordered lists\;
    $L' := |P|$-element array of ordered lists\;
    $baseline := 0$\;
    $baseline' := 0$\;
    \ForEach{child $y$ of $x$ in $T$}{
      $b := domset(G,T,y,\nothing,D)$\;
      $baseline := baseline + b$\;
      $b' := domset(G,T,y,\nothing,D \cup \{x\})$\;
      $baseline' := baseline + b'$\;
      \For{$S_i \in \{S_1,\dots,S_l\}$}{
        $c := domset(G,T,y,S_i,D) - b$\;
        $c' := domset(G,T,y,S_i,D \cup \{x\}) - b'$\;
        Insert $(c,y)$ into ordered list $L[i]$ and keep only smallest $l$ elements\;
        Insert $(c',y)$ into ordered list $L'[i]$ and keep only smallest $l$ elements\;
      }
    }
    \tcc{Find minimal cost of dominating $\{S_1,\dots,S_l\}$ from $L$ and $L'$ by 
      e.g.~solving appropriate matching problems.}
    result := $min(result,find\_min\_solution(L) + baseline)$\label{line:min-1}\;
    result := $min(result,find\_min\_solution(L') + baseline' + 1)$\label{line:min-2}\;
  }
  \Return result\;
\end{algorithm}

\begin{lemma}
  Algorithm~\ref{alg:domset-polyspace} runs in time $t^{O(t^2)} \cdot n$.
\end{lemma}
\begin{proof}
  The running time when $x$ is a leaf is bounded by $O(t^2)$, since
  all operations exclusively involve some subset of the $t$ nodes in
  $P_x \cup \{x\}$. Since $|P| \leq t$ the number of partitions of $P$
  is bounded by $t^t$. When $x$ is not a leaf the only time spent on
  computations which are not recursive calls of the algorithm are all
  trivially bounded by $O(t)$, except the time spent on
  $find\_min\_solution$, which can be solved via a matching problem in time~$\poly(t)$. The number of recursive calls that a single
  call on a node $x$ makes on a child $y$ is $O(t \cdot
  t^t)$ which bounds total number of calls on a single node by
  $t^{O(t^2)}$. This proves the claim.
\end{proof}

\begin{lemma}
  Algorithm~\ref{alg:domset-polyspace} uses $O(t^3 \log t + t
  \log n)$ bits of space.
\end{lemma}
\begin{proof}
  There are at most $t$ recursive calls on the stack at any point. We
  will show that the space used by one is bounded by $O(t^2 \log t +
  \log n)$. Each call uses~$O(t)$ sets, all of which have size
  at most $t$. The elements contained in these sets can be represented by
  their position in the path to the root of $T$, thus they use
  at most $O(t^2 \log t)$ space. The arrays of ordered lists~$L,L'$ contain at
  most $t^2$ elements and all entries are $\leq t$ or $\infty$:
  If the additional cost (compared to the baseline cost) of 
  dominating a block~$S_i$ of the current partition from
  some subtree~$T_y$ exceeds~$|S_i| \leq t$, we can disregard this 
  possibility---it would be cheaper to just take all vertices in~$S_i$,
  a possibility explored in a different branch. 

  To find a minimal solution from the table we need to avoid using
  the same subtree to dominate more than one element of the
  partition; however, at any given moment we only need to distinguish
  at most~$t^2$ subtrees. Thus the size of the arrays $L$ and $L'$ is bounded
  by $O(t^2 \log t)$. The only other space consumption is caused by a
  constant number of variables ($result$, $baseline$, $baseline'$,
  $b$, $b'$ and $x$) all of them $\leq n$. Thus the space consumption
  of a single call is bounded by $O(t^2 \log t + \log n)$ and the
  lemma follows.
\end{proof}

\subsection{Fast \name{Dominating Set} using \boldmath$O(2^t t \log t + t
  \log n)$ space}

We have seen that it is possible to solve \name{Dominating Set} on low-treedepth
graphs in a space-efficient manner. However, we traded exponential space
against superexponential running time and it is natural to ask whether there
is some middle ground. We present Algorithm~\ref{alg:domset-2t-space} to
answer this question: its running time~$O(3^t \log t \cdot n)$ is competitive
with the default dynamic programming but its space complexity~$O(2^t \log t + t \log n)$
is exponentially better. The basic
idea is to again branch from the top deciding if the current node~$x$ is in the
dominating set or not. Intertwined in this branching we compute
a function which for a subtree $T_x$  and a set $S \subseteq P_x$
gives the cost of dominating $V(T_x) \cup S$ from $T_x$. For each recursive
call on a node we only need this function for subsets of $P_x$ which are not dominated. If~$d$ is the number of nodes of~$P_x$ that are currently
in the dominating set, the function only needs to be computed for~$2^{t-d}$ sets. This allows us to keep
the running time of $O^*(3^t)$, since $\sum_{i=0}^{t} {t \choose i} \cdot 2^{t
- i} = 3^t$, while only creating tables with at most $O(2^t)$ entries. By
representing all values in these tables as $\leq t$ offsets from a base value,
the space bound $O(2^t t \log t + t \log n)$ follows. For two
functions $M_1, M_2$ with domain~$2^U$ for some ground-set~$U$ (realized via
associative arrays in the algorithm) we use the notation~$M_1 \ast M_2$ to
denote the convolution
$
  (M_1 \ast M_2)[X] := \min_{A \uplus B = X} M_1[A] + M_2[B], 
$
for all~$X \subseteq U$.

\begin{algorithm}
  \footnotesize
  \caption{domset{\label{alg:domset-2t-space}}}
  \KwIn{A graph $G$ and a treedepth decomposition $T$ of $G$, a node
    $x$ of $T$ and a set $D \subseteq V(G)$.}  \KwOut{The size of a
    minimum Dominating Set.} \BlankLine

  $M,M_1,M_2 :=$ are empty associative arrays. If a set is not in the 
  array its value is $\infty$\;
  \If{$x$ is a leaf in $T$\label{line:start-leaf}}{
    $M[N_G[x] \setminus D] := 1$\;
    \lIf{$x \in N_G[D]$}{
      $M[\varnothing] := 0$
    }
    \Return $M$\;\label{line:end-leaf}
  }

  \tcc{Assume the children of $x$ are $\{y_1,\dots,y_l\}$.}
  \For{$i \in \{1,\dots,l\}$\label{line:start-x-notin-d}}{
    $M' := domset(G,T,y_i,D)$\;
    $M_1 := M_1 \ast M'$\;
  }
  \tcc{$x$ is not in the dominating set. 
    Discard entries where $x$ is undominated.}
  \lIf{$x \notin N_G[D]$}{
      delete all entries $S$ from $M$ where $x \notin S$\label{line:end-x-notin-d}
  }
  \For{$i \in \{1,\dots,l\}$\label{line:start-x-in-d}}{
    $M' := domset(G,T,y_i,D \cup \{x\})$\;
    $M_2 := M_2 \ast M'$\;
  }
  \lForEach{$S \in M_2$}{
    $M[S] := M[S] + 1$\label{line:end-x-in-d}
  }
  $M := M_1 \ast M_2$\;
  \tcc{Forget $x$.}
  \lForEach{$S \in M$ where $x \notin S$\label{line:start-remove-x-info}}{
    $M[S] := min(M[S],M[S \cup \{x\}])$
  }
  Delete all entries $S$ from $M$ where $x \in S$\;\label{line:end-remove-x-info}
  \lIf{$x$ is the root of $T$}{
    \Return $M[\varnothing]$
  }
  \lElse{
    \Return $M$
  }
\end{algorithm}

\begin{theorem}
  For a graph $G$ with treedepth decomposition $T$,
  Algorithm~\ref{alg:domset-2t-space} finds the size of a minimum
  dominating set in time $O(3^t \log t \cdot n)$ using $O(2^t t
  \log t + t \log n)$ bits of space.
\end{theorem}

\noindent
We divide the proof into lemmas as before.

\begin{lemma}
  Algorithm~\ref{alg:domset-2t-space} called on $G,T,r,\emptyset$,
  where~$T$ is a treedepth decomposition of~$G$ with root~$r$,
  returns the size of a minimum dominating set of $G$.
\end{lemma}
\begin{proof}
  Notice that the associative array $M$ represents a function which
  maps subsets of $P_x \setminus D$ to integers and $\infty$. At the
  end of any recursive call, $M[S]$ for $S \subseteq P_x \setminus D$
  should be the size of a minimal dominating set in $T_x$ which
  dominates $T_x$ and $S$ assuming that the nodes in $D$ are part of
  the dominating set. We will prove this inductively.

  Assume $x$ is a leaf. We can always take $x$ into the
  dominating set at cost one. In case $x$ is already dominated we have
  the option of not taking it, dominating nothing at zero cost. This
  is exactly what is computed in lines
  \ref{line:start-leaf}--\ref{line:end-leaf}.

  Assume now that $x$ is an internal, non-root node of $T$. First, in lines
  \ref{line:start-x-notin-d}--\ref{line:end-x-notin-d} we assume that
  $x$ is not in the dominating set. By inductive assumption calling
  $domset$ on a child $y$ of $x$ returns a table which contains the
  cost of dominating $T_y$ and some set $S \subseteq P_y \setminus
  D$. By merging them all together $M_1$ represents a function which
  gives the cost of dominating some set $S \subseteq (P_x \cup \{x\})
  \setminus D$ and all subtrees rooted at children of $x$. We just
  need to take care that $x$ is dominated. If $x$ is not dominated by
  a node in $D$, then it must be dominated from one of the
  subtrees. Thus we are only allowed to retain solutions which
  dominate $x$ from the subtrees. We take care of this on
  line~\ref{line:end-x-notin-d}. After this $M_1$ represents a
  function which gives the cost of dominating some set $S \subseteq
  (P_x \cup \{x\}) \setminus D$ and $T_x$ assuming $x$ is not in the
  dominating set. Then we compute a solution assuming $x$ is in the
  dominating set in lines
  \ref{line:start-x-in-d}--\ref{line:end-x-in-d}. We first merge the
  results on calls to the children of $x$ via convolution. Since we
  took $x$ into the dominating set we increase the cost of all entries
  by one. After this $M_2$ represents the function which gives the
  cost of dominating some set $S \subseteq P_x \setminus D$ and $T_x$
  assuming $x$ is in the dominating set. We finally merge $M_1$ and
  $M_2$ together with the min-sum convolution.  Since we have taken
  care that all solutions represented by entries in $M$ dominate $x$
  we can remove all information about $x$. We do this in lines
  \ref{line:start-remove-x-info}--\ref{line:end-remove-x-info}. Finally,
  $M$ represents the desired function and we return it. When $x$ is
  the root, instead of returning the table we return the value for the
  only entry in $M$, which is precisely the size of a minimum
  dominating set of $G$.\looseness-1
\end{proof}

To prove the running time of Algorithm~\ref{alg:domset-2t-space} we will need
the values $M$ to be all smaller or equal to the depth of $T$. Thus we first
proof the space upper bound. In the following we treat the
associative arrays $M$, $M_1$ and $M_2$ as if the entries where values
between $0$ and $n$. We will show that we can represent all values as
an offset $\leq t$ of a single single value between $0$ and $n$.

\begin{lemma}\label{lemma:2t-space}
  Algorithm~\ref{alg:domset-2t-space} uses $O(2^t t \log t + t
  \log n)$ bits of space.
\end{lemma}
\begin{proof}
  Let $t$ be the depth of the provided treedepth decomposition. It is
  clear that the depth of the recursion is at most $t$. Any call to
  the function keeps a constant number of associative arrays and nodes
  of the graph in memory. By construction these associative arrays
  have at most $2^t$ entries. For any of the computed arrays $M$ the
  value of $M[\varnothing]$ and $M[S]$ for any $S \neq \varnothing$
  can only differ by at most $t$. We can thus represent every entry
  for such a set $S$ as an offset from $M[\varnothing]$ and use $O(2^t
  \log t + \log n)$ space for the tables. This together with the bound
  on the recursion depth gives the bound $O(2^t t \log t + t \log n)$.
\end{proof}

\begin{lemma}
  Algorithm~\ref{alg:domset-2t-space} runs in time $O(3^t \log t \cdot
  n)$.
\end{lemma}
\begin{proof}
  On a call on which $d$ nodes of $P_x$ are in the dominating set the
  associative arrays have at most $2^s$ entries for $s = t - d$. As
  shown above the entries in the arrays are $\leq s$ (except one).
  Hence, we can use fast subset convolution to merge the arrays in
  time $O(2^s \log s)$~\cite{bjorklund2007fourier}. It follows that
  the total running time is bounded by $O\big(n \cdot \sum_{i=0}^{t}
  {t \choose i} \cdot 2^{t - i} \log (t - i) \big) = O(3^t \log t
  \cdot n)$.
\end{proof}


%% file: conclusion.tex
\section{Conclusion and future work}\label{sec:conclusion}

We have shown that single-pass dynamic programming on treedepth, tree or
path decompositions without preprocessing of the input 
must use space exponential in the width/depth,
confirming a common suspicion and proving it rigorously for the first
time. This complements previous SETH-based arguments about the
running time of arbitrary algorithms on low treewidth graphs.
We further demonstrate that treedepth allows non-DP algorithms
that use only polynomial space in the depth of the provided
decomposition. Both our lower bounds and the provided algorithm
for \name{Dominating Set} appear as if they could be special cases of
a general theory to be developed in future work and we further ask
whether our result can be extended to less stringent definitions
of `dynamic programming algorithms'.
\looseness-1

Despite the less-than-ideal theoretical bounds of the
presented \name{Dominating Set} algorithms, the opportunities for
heuristic improvements are not to be slighted. Take the pure
branching algorithm presented in Section~\ref{sec:branching}.
During the branching procedure, we generate all partitions from the
root-path starting at the current vertex. However, we actually only
have to partition those vertices that are not dominated yet (by virtue
of being themselves in the dominating set or being dominated by another
vertex on the root-path). A sensible heuristic as to which branch---including
the current vertex in the dominating set or not---to explore first,
together with a \emph{branch \&{} bound} routine should keep us from
generating partitions of very large sets.
A similar logic applies to the mixed dynamic
programming/branching algorithm since the tables only have to contain
information about sets that are not yet dominated. The tables could thus
be kept a lot smaller than their theoretical bounds indicate.

Furthermore, it seems reasonable that in practical settings, the nodes
near the root of treedepth decompositions are more likely to be part of a
minimal dominating set. If this is true, computing a treedepth decomposition
would serve as a form of smart preprocessing for the branching, a rough `plan
of attack', if you will. How much such a \emph{guided branching} improves upon
known branching algorithms in practice is an interesting avenue for further
research.\looseness-1


It is still an open question whether \name{Dominating Set} can
be solved in time $(3-\epsilon)^t \cdot \poly(n)$ when
parameterized by treedepth. Our lower bound result implies that if
such an algorithm exists, it cannot be purely dynamic programming.
